\documentclass[12pt]{article}
\usepackage{amsmath, amsthm, amssymb}
\usepackage{enumerate}
\usepackage{endnotes}
\usepackage{verbatim} 

\newtheorem{theorem}{Theorem}
\newtheorem{lemma}[theorem]{Lemma}
\newtheorem{claim}[theorem]{Claim}
\newtheorem{defi}[theorem]{Definition}
\newtheorem{cor}[theorem]{Corollary}

\newcommand{\e}{\varepsilon}
\newcommand{\eps}{\varepsilon}

\DeclareMathOperator{\diam}{diam}
\DeclareMathOperator{\E}{E}

\newcommand{\oN}{{\overline{N}}}
\newcommand{\N}{{\mathbb{N}}}

\newcommand{\NN}{{\mathbb{N}}}

\newcommand{\RR}{{\mathbb{R}}}

\newcommand{\menge}[2]{\left\{{#1}\ \left| \ {#2}\right.\right\}}
\newcommand{\thmref}[1]{Theorem~\ref{thm:#1}}
\newcommand{\clref}[1]{Claim~\ref{cl:#1}}

\newcommand*\samethanks[1][\value{footnote}]{\footnotemark[#1]}

\begin{document}

\title{Strong Robustness of Randomized\\ Rumor Spreading Protocols\thanks{A short version of this work appeared in the proceedings of the 20th International Symposium on Algorithms and Computation (ISAAC 2009).}}

\author{Benjamin Doerr\thanks{Partially supported by the German Science
Foundation (DFG) via its priority program ``SPP 1307: Algorithm Engineering'', grant DO 479/4-2.}\\
Max-Planck-Institut f\"ur Informatik\\
Saarbr\"ucken, Germany\\
\and Anna Huber\thanks{Work done while the authors were with the Max-Planck-Institut f\"ur Informatik.}\\
Durham University, United Kingdom\\
\and Ariel Levavi\samethanks\hspace{6pt}\thanks{Supported by the Study Scholarship awarded by the German Academic Exchange Service (DAAD).}\\
University of California, San Diego\\
}
\date{}

\maketitle

\begin{abstract}
  Randomized rumor spreading is a classical protocol to disseminate information across a network. At SODA 2008, a quasirandom version of this protocol was proposed and competitive bounds for its run-time were proven. This prompts the question: to what extent does the quasirandom protocol inherit the second principal advantage of randomized rumor spreading, namely robustness against transmission failures? 
  
In this paper, we present a result precise up to $(1 \pm o(1))$ factors. We limit ourselves to the network in which every two vertices are connected by a direct link.  Run-times accurate to their leading constants are unknown for all other non-trivial networks. 
  
  We show that if each transmission reaches its destination with a probability of $p \in (0,1]$, after $(1+\e)\left(\frac{1}{\log_2(1+p)}\log_2n+\frac{1}{p}\ln n\right)$ rounds the quasirandom protocol has informed all $n$ nodes in the network with probability at least $1-n^{-p\e/40}$. Note that this is faster than the intuitively natural $1/p$ factor increase over the run-time of approximately $\log_2 n  + \ln n $ for the non-corrupted case. 
  
  We also provide a corresponding lower bound for the classical model. This demonstrates that the quasirandom model is at least as robust as the fully random model despite the greatly reduced degree of independent randomness. 
\end{abstract}

{\sloppy
\section{Introduction}
  Disseminating information in a network, that is, making information that is known to a single node available to all other nodes, is a classical problem. A simple, yet powerful approach is \emph{randomized rumor spreading}, also known as \emph{random phone calls}. In this setting, each node that is already informed participates in the dissemination process by  randomly calling neighbors and passing along copies of the information. Besides being self-organized, this approach has two crucial advantages. (i) It is fast. For many important network topologies, $O(\log n)$ rounds suffice to inform all $n$ nodes with high probability. (ii) It is robust against transmission failure. Often, a constant fraction of independently chosen transmission failures does not cause serious problems, but merely increases the time required by a constant factor. 
  
  Success of the basic randomized rumor spreading protocol has motivated the study of several variants. In \cite{DFS08}, a quasirandom version of the protocol was proposed. This version is structurally simpler, uses less randomness, and is especially beneficial in that each vertex contacts each neighbor at most once. Nonetheless, most run-time guarantees known for the classical model still hold for the quasirandom version, some in an even stronger form. However, little was previously known about the robustness of this model. 
  
  In this paper, we offer a detailed investigation of the robustness of the quasirandom protocol. We use the following model of lossy communication, which was analyzed in \cite{MR2239598}. We assume that each message reaches its target with a certain probability $p \in (0,1)$ independently for all transmissions. For networks in which every two nodes are connected by a link, we demonstrate that this lossiness only increases the run-time by a small constant factor, which we precisely determine for each value of $p$. Surprisingly, it is smaller than $1/p$. For example, for $p = 1/2$, the run-time increases by a factor of $1.828$ only. The same result also holds for the classical, fully random model.
  
  In addition, we show that the corresponding slow-down for the classical model is at least this factor. From this we conclude that the quasirandom model is at least as robust as the classical model. 
  
  This is the first time, for the classical as well as for the quasirandom model, that results precise up to the leading constant are shown for robustness. 
  
\subsection{Randomized Rumor Spreading}

  The classical, fully random \emph{randomized rumor spreading} protocol was first investigated by Frieze and Grimmett~\cite{FG85}. They proposed the following model. Let $G = (V, E)$ be an undirected graph. At the start of the protocol, a single vertex $s \in V$ knows a piece of information that is to be disseminated to all other vertices. We say that $s$ is \emph{informed}. The protocol proceeds in rounds (hence it assumes a common clock). In each round, every informed vertex $v$ chooses a neighbor $u_v \in N(v) := \{u \in V \mid \{u,v\} \in E\}$ uniformly at random and sends a copy of the information to it. This results in $u_v$ becoming informed, if it is not already, and in $u_v$ participating in the dissemination process in subsequent rounds. This process defines a random variable $T_s$, which denotes the number of rounds after which all vertices in the network are informed, assuming that the initially informed vertex is $s$. The \emph{broadcast time} $T$ is then defined as the maximum of all $T_s$, $s \in V$, where the $T_s$ are defined over independent probability spaces.
  
  Frieze and Grimmett demonstrate that if the network is a complete graph on $n$ vertices, the broadcast time satisfies $T = (1 \pm o(1)) (\log_2 n + \ln n)$ with probability $1 - o(1)$. For hypercubes and random graphs $G(n,p)$, where $p \ge (1 + \eps) \ln(n)/n$, Feige, Peleg, Raghavan and Upfal~\cite{FPRU90} also determine a broadcast time of $\Theta(\log n)$ with probability $1 - 1/n$, albeit without making the implicit constant precise. They also provide the general bounds of $12 n \log n$ and $O(\Delta(G) (\diam(G)+ \log n))$ for arbitrary $n$-vertex graphs. Subsequent work by Sauerwald~\cite{S07}, Els\"asser and Sauerwald~\cite{ES07} and Berenbrink, Els\"asser and Friedetzky~\cite{BEF08} shows that the $\Theta(\log n)$ bound also holds for expander graphs, Cayley graphs and random regular graphs.
  
  We shall not extensively discuss the practical side of the randomized rumor spreading protocol. We refer the interested reader to the aforementioned paper~\cite{FPRU90} as well as the paper by Karp, Shenker, Schindelhauer and V\"ocking~\cite{KSSV00} for a general discussion, or the works of Demers et al.~\cite{DGH+88} and Kempe, Dobra and Gehrke~\cite{KDG03} for particular applications. What are generally recognized as the three key advantages of randomized rumor spreading are speed (logarithmic broadcast time on important network topologies); self-organization (there is no central authority involved); and robustness against transmission failure. Contrary to broadcast times, significantly less work has been done to quantify the robustness of the randomized rumor spreading protocol.

    As far as we are aware, the only results on the robustness of randomized rumor spreading are due to Els\"asser and Sauerwald~\cite{ES06full}. They consider the model where each transmission does not reach its destination with failure probability $1-p$, and failures in different rounds are independent of each other. They assert that the broadcast time for all graphs in this lossy model is at most a factor of $O(1/p)$ larger than in the model without transmission failures.
    
\subsection{Quasirandom Rumor Spreading}

  The above results show that randomized rumor spreading is a very powerful approach to dissemination problems. However, taking all decisions independently at random also has some unwanted effects. For example, a vertex may contact one of its neighbors twice before contacting all of its other neighbors. This may only be a minor problem for dense graphs like the complete graph, but for sparse graphs, it may increase the broadcast time significantly. 
  
  Let $G$ be a star on $n$ vertices, i.e., a graph with one central vertex such that all other vertices have this vertex as their only neighbor. Clearly, any dissemination process where each vertex can send out at most one transmission per round needs at least $n-1$ rounds, simply because the central vertex has $n-1$ neighbors that cannot be informed by other vertices. However, due to the coupon collector effect, randomized rumor spreading on the star needs $\Theta(n \log n)$ rounds with probability $1 - o(1)$. 
  
  This type of imbalance in informing one's neighbors may be avoided by choosing the destination of the current transmission uniformly at random from those neighbors which have not yet been contacted by the originating vertex. However, this requires tracking all previously sent messages and is therefore less desirable. Motivated by the paradigm of quasirandomness, Friedrich, Sauerwald and the first author~\cite{DFS08} suggest the following \emph{quasirandom rumor spreading} protocol.
  
  In this model, each vertex is equipped with a cyclic permutation of its neighbors. As before, the protocol proceeds in rounds, and all informed vertices participate in the dissemination process. However, each vertex only directs its first transmission to a random neighbor. Subsequently, it informs the successors of the first addressee on its list. We shall not make any assumptions about the structures of these cyclic lists.

Before analyzing the quasirandom protocol, let us discuss it from an implementation
point of view. From a theory perspective, we immediately note that the quasirandom model requires each
vertex to store the permutation of its neighbors, which may utilize up to $\Theta(n \log n)$ bits.
This is not necessary for the fully random model. However, we may assume that in
most networks each vertex already has some list or array of its neighbors, since the
information regarding how to contact a neighbor must be stored somewhere. In this scenario, the use of the lists does not increase the complexity. Rather, it appears that the
quasirandom protocol needs less resources. In particular, it requires significantly fewer random
bits. This is beneficial if we consider randomness costly, and useful if we want to trace an
actual run of the protocol.

The core question to be answered is whether the quasirandom protocol
works well even if we are not permitted to design the lists. Surprisingly, the answer
is yes.

For all lists that can be present at each vertex, $O(\log n)$ rounds suffice with high probability to inform all the vertices of a complete graph $K_n$, a hypercube $Q_n$, an expander graph on $n$ vertices (some extra conditions are needed here), or a random graph $G(n,p)$ with $p \ge (1 + \eps) (\ln n)/n$~\cite{DFS08,DFS09}. Naturally, the lower bound of $\log_2n$ rounds valid for the fully random model also holds for the quasirandom model. Once again these bounds fall within the right order of magnitude. 

Sharper bounds analogous to those defined by Frieze and Grimmett are known for the complete graph.
In~\cite{TightBounds}, it is shown that with probability $1 - o(1)$, the number of rounds needed to inform all vertices is $(1 \pm o(1))(\log_2 n  + \ln n )$. 

In some settings, we observe better broadcast times than in the classical model. One example is the random graph with edge probability $p = (\ln n + \omega(1))/n$ only minimally above the connectivity threshold. Nevertheless, with probability $1 - o(1)$, the random graph is such that with high probability the quasirandom protocol needs only $O(\log n)$ rounds independent of the starting point. This is a notable advantage over the fully random model. Feige et al.~\cite{FPRU90} demonstrate that for $p = (\ln n  + O(\log\log n))/n$, the random graph with probability $1 - o(1)$ is such that $\Theta(\log^2 n )$ rounds are necessary to spread the rumor with high probability. 

The bounds obtained for arbitrary graphs are also superior for the quasirandom model. For the fully random model, we saw above that $12 n \ln n$ and $O(\Delta(G) (\diam(G) + \log n ))$ rounds suffice to inform all vertices of an $n$-vertex graph $G$ with high probability~\cite{FPRU90}. For the quasirandom model, it is easily proven that after $2n -3$ or $\Delta(G) \diam(G)$ rounds, all vertices are informed with probability one. 

\subsection{Robustness of the Quasirandom Protocol}

The above results show that the broadcast time of the quasirandom rumor spreading protocol is quite well understood. Together with the experimental investigation~\cite{DFKS09}, all results indicate that the quasirandom protocol achieves comparable or better broadcast times than the random model. For the equally important aspect of robustness, much less is known. Since it would typically seem that robustness of randomized algorithms is caused by the large number of independent random decisions taken by the algorithm, one may conclude that the quasirandom protocol is less robust. 

The experimental evaluation in~\cite{DFKS09} debunks this assertion. For both the hypercube and the complete graph on $2^{12}$ vertices, it was observed that if messages sent across the network using either protocol get lost with probability $\frac 12$, the broadcast time increases by a factor of between $1.8$ and $1.9$. 

The only theoretical result pertaining to robustness is the one in~\cite{DFS09}. Let $G$ be a graph, $T \in \N$ and $\gamma \ge 1$ such that the quasirandom protocol independent of the starting vertex with probability $1 - n^{-\gamma}$ succeeds in informing all other nodes within $T$ rounds. Then in the presence of transmission failures (independently chosen with probability $1-p$), independent of the starting vertex, $4 \gamma (1/p) T$ rounds of a modified quasirandom protocol suffice to inform all vertices with probability $1 - 2n^{-\gamma}$. The modification of the protocol needed to prove this result is that the recipient of a message returns a feedback message to the sender (which also gets lost with probability $p$). Whenever the sender does not receive a feedback message, he tries to reach the same addressee in the next round. With this modification, however, the result is slightly weaker, in particular, because the feedback modification makes the protocol significantly less simple.

In addition, the robustness result in~\cite{DFS09} leaves room for constant factor differences between the random and the quasirandom models in the presence of transmission faults. 

To gain a deeper understanding, we study the robustness of quasirandom rumor spreading (without the feedback modification) on the complete graph in this work. Recall that the complete graph is the only graph for which a broadcast time of one of the two models precise up to 
the leading constant is published. We show the following main result.

\emph{Main Result:} For all $\eps > 0$ and $p \in (0,1]$ the quasirandom rumor spreading protocol with arbitrary lists, despite independent message losses occurring with probability $1-p$, succeeds with probability at least $1 - n^{-p\e/40}$ in  informing all other vertices from a given vertex in time at most $(1 + \eps)(\log_{1+p} n + \frac 1p \ln n)$.

This result is interesting for two reasons. Firstly, it shows that the quasirandom protocol is even more robust than previous results indicate. Note that the above bound is strictly better than $(1/p)(1+o(1))(\log_2 n + \ln n)$, that is, $(1/p)$ times the bound for the case without faulty transmissions.

Secondly, our results imply that the quasirandom protocol is at least as robust as the classical one. To prove this, we show a corresponding lower bound for the fully random protocol. 

We should add that our proof for the upper bound of the quasirandom model can be modified to yield a corresponding proof for the classical protocol. This is the first bound to make the robustness of the classical protocol precise up to the leading constant. 

\subsection{Organization of the Paper}
In Section \ref{sec:tools}, we collect some
well-known probabilistic tools we require for the remainder of the
paper. We proceed in Section 3 by analyzing the randomized rumor spreading model, as this is shorter and easier and may serve as an introduction to the techniques we use. We provide a lower bound on the broadcast time. We then analyze the quasirandom version of this model in Section 4. This is the main part of our proof. Our goal is to show that the rumor is disseminated in the quasirandom model at least as quickly as it is in the random model, and so in this section we will focus on determining an upper bound on the broadcast time.

\section{Probabilistic Tools}\label{sec:tools}
A basic tool that we will use in the following proofs is the Chernoff bound~\cite{Chernoff}. 
This provides exponentially small bounds for the probability that a binomially 
distributed random variable deviates significantly from its expected value. 
This classical result can be found for example in~\cite{MU05} in the following form.
\begin{theorem}[Chernoff Bounds]
\label{thm:chernoff}
Let $X_1, \dots ,X_n$ be independent random variables, taking
values in $\{0, 1\}$. Let $\displaystyle X := \sum_{i=1}^n X_i$ and let $\delta \in (0,1]$. Then 
$$ \Pr\left(X \leq (1-\delta)\E (X)\right)\leq e^{-\delta^2\E (X)/2},$$
and
$$ \Pr\left(X \geq (1+\delta)\E (X)\right)\leq e^{-\delta^2\E (X)/3}.$$
\end{theorem}
For random variables which take more than two values, but are still independent and bounded, Theorem 2 from \cite{Hoe63} yields directly the following upper and lower tail bounds.
\begin{theorem}[Hoeffding Bounds]\label{thm:Hoeffding}
Let $X_1, \dots ,X_n$ be independent random variables, and for every $i \in [n] := \NN_{\leq n}$ let $a_i, b_i \in \RR$ such that $0 \leq a_i < b_i$ and $X_i$ takes
values in $[a_i, b_i]$. Let $\displaystyle X := \sum_{i=1}^n X_i$ and let $\delta > 0$. Then 
$$ \Pr\left(X \leq (1-\delta)\E (X)\right)\leq e^{-\frac{2\delta^2\E (X)^2}{\sum_{i=1}^n(b_i-a_i)^2}},$$
 and
$$ \Pr\left(X \geq (1+\delta)\E (X)\right)\leq e^{-\frac{2\delta^2\E (X)^2}{\sum_{i=1}^n(b_i-a_i)^2}}.$$
\end{theorem}
There are places where we would like to use Chernoff bounds, but we do not have independence of the random variables. Here results of Panconesi and Srinivasan~\cite{PS97,SrinFOCS2001} show that we may use the classical Chernoff bounds even under the more general assumption that the random variables are negatively correlated. This is defined as follows.

\begin{defi}
The random variables $X_1, \dots ,X_n$, taking
values in $\{0, 1\}$, are called {\emph{negatively correlated}}, if for every subset $I \subseteq \{1,\dots , n\}$ we have 
$$\Pr{\left(\bigwedge_{i \in I}X_i = 1\right)} \leq \prod_{i \in I}\Pr{\left(X_i = 1\right)},$$
and
$$\Pr{\left(\bigwedge_{i \in I}X_i = 0\right)} \leq \prod_{i \in I}\Pr{\left(X_i = 0\right)}.$$
\end{defi}
The results of Panconesi and Srinivasan~\cite{PS97,SrinFOCS2001} yield the following lemma.

\begin{lemma}[Chernoff Bounds for negatively correlated random variables]
\label{thm:chernoffnc}
Let $X_1, \dots ,X_n$ be negatively correlated random variables, taking
values in $\{0, 1\}$. Let $X := \sum_{i=1}^n X_i$ and let $\delta \in (0,1]$. Then 
$$ \Pr \left(X \leq (1-\delta)\E (X)\right)\leq e^{-\delta^2\E (X)/2},$$
and
$$ \Pr \left(X \geq (1+\delta)\E (X)\right)\leq e^{-\delta^2\E (X)/3}.$$
\end{lemma}

Later in this paper we will apply the inequality by Azuma \cite{Azuma}. Intuitively, it provides strong bounds on the probability that a function defined on a set of independent random variables deviates significantly from its expectation, when the value of the function is affected only slightly by changes to only one of its arguments. 
We will use it in the following version, stated in \cite[Lemma 1.2]{mcdiarmid89}.

\begin{lemma}[Azuma-Inequality]
\label{thm:azuma} Let $X_1, \dots ,X_n$ be independent random variables, with $X_i$ taking
values in a set $\Omega_i$ for each $i$. Suppose that the (measurable) function $f:\prod_{i=1}^n\Omega_i\rightarrow\RR$ satisfies
\begin{equation*}
|f(x) - f(x')| \leq c_i
\end{equation*}
whenever the vectors $x$ and $x'$ differ only in the $i$th coordinate. Let $Y$ be the random variable
$f(X_1, \dots ,X_n)$. Then for any $t > 0$,
\begin{equation*}
 \Pr (|Y - \E(Y)| \geq t) \leq 2 \exp\left(-\frac{2t^2}{\sum_{i=1}^nc_i^2}	\right).
\end{equation*}
\end{lemma}

\section{Lower Bound for Randomized Rumor Spreading}\label{sec:lower}

In this section, we analyze the classical (fully random) rumor spreading model in which each informed node randomly chooses a neighbor to inform at the beginning of each round, but only makes successful contact with probability~$p \in (0,1]$. We prove the following lower bound for the broadcast time.

\begin{theorem}
 Let $\e>0$ and~$p \in (0,1]$. With probability $1-e^{-\Omega(n^{\e/6})}$, the number of rounds we need to inform all the nodes of the complete graph on $n$ vertices using the random rumor spreading protocol with message success probability $p$ is at least

\begin{equation*}
(1-\e)\left(\log_{1+p}n+\tfrac{1}{p}\ln n\right).
\end{equation*}

\end{theorem}

The key to this proof is to split up the rumor spreading process into three phases. The first phase is composed of the rounds that occur between the start of the process and the end of the first round at which point $n^{\e/2}$ nodes are informed. The second phase begins directly after Phase 1 terminates, and continues until the end of the first round after which $n/4$ nodes are informed. Within each round of Phase 2, the number of informed nodes will grow by a multiplicative factor. The last phase begins directly after Phase 2 terminates, and continues until all the nodes are informed. In this phase we observe a type of coupon collector process.

In order to establish the lower bound posited above, we give lower bounds for the durations of Phases 2 and 3. 
By $N_t$ we will denote the set of vertices that are newly informed at a given time-step $t$. 
By $I_t$ we will denote the set of vertices informed by time $t$.

\begin{lemma}\label{RandomPhase2}
Let $\e > 0$. With probability $1-e^{-\Omega(n^{\e/6})}$, we need more than $(1-\e)\log_{1+p}n$ rounds to complete Phase 2.
\end{lemma}

\begin{proof}
Let $t_1$ denote the number of rounds needed to inform the first $n^{\e/2}$ nodes. Note that this means that $n^{\e/2} \leq |I_{t_1}| < 2n^{\e/2}$. Let $t \geq t_1$. We have $\E(|N_{t+1}|)\leq p|I_t|$. Enumerate the nodes in $I_t$ from 1 to $|I_t|$, and define the indicator random variables $X_1,\dots X_{|I_t|}$ such that

\begin{equation*}
X_i = \begin{cases}
       1 &\mbox{if vertex }i\mbox{ has successfully contacted another vertex}\\
       0 &\mbox{otherwise.}
      \end{cases}
\end{equation*}

In this context a successful contact refers only to the transmission of the rumor, regardless of whether or not the contacted vertex was already informed or is also contacted by another vertex. 
Therefore the random variables $X_1,\dots , X_{|I_t|}$ are independent.

If $X:=\sum_iX_i$, then $\E(X)=p|I_t|$. 
It is intuitive that $|N_{t+1}|\leq X$, because $X$ not only counts all the nodes in $N_{t+1}$, but also counts nodes multiple times if they are contacted by multiple nodes, and counts nodes that are contacted in round $t+1$ that have already been informed in  previous rounds. 
Therefore, any upper bound we can find on the size of $X$ also holds as an upper bound for the size of $N_{t+1}$. 
But because of the independence of $X_1,\dots , X_{|I_t|}$, the random variable $X$ is a lot easier to handle.

Using Chernoff bounds, we see that

\begin{align*}
\Pr\left(X>(1+n^{-\e/6})\E(X)	\right)
& \leq\exp\left(-\tfrac{1}{3}n^{-\e/3}p|I_t|\right)\\
& \leq\exp\left(-\tfrac{1}{3}n^{-\e/3}pn^{\e/2}\right)\\
& =\exp\left(-\tfrac{1}{3}pn^{\e/6}\right)\\
& =e^{-\Omega(n^{\e/6})}.
\end{align*}
So with probability $1-e^{-\Omega(n^{\e/6})}$, the number of nodes informed after $t+1$ rounds satisfies $$|I_{t+1}|\leq |I_t|+(1+n^{-\e/6})p|I_t|\leq (1+n^{-\e/6})(1+p)|I_t|.$$ We can therefore infer by using recursion that for every $k\in \N$ we have

\begin{equation*}
|I_{t+k}|\leq (1+n^{-\e/6})^k(1+p)^k|I_t|
\end{equation*}
with probability $1-ke^{-\Omega(n^{\e/6})}$.
Pick $k = (1-\e)\log_{1+p}n$. Under the assumption that $n$ is sufficiently large, we compute

\begin{align*}
 |I_{t_1+(1-\e)\log_{1+p}n}|&\leq (1+n^{-\e/6})^{(1-\e)\log_{1+p}n}(1+p)^{(1-\e)\log_{1+p}n}2n^{\e/2}\\
&=(1+n^{-\e/6})^{(1-\e)\log_{1+p}n}n^{1-\e}2n^{\e/2}\\
&\leq \exp\left(n^{-\e/6}(1-\e)\log_{1+p}n\right)2n^{1-\e/2}\\
&\leq 4n^{1-\e/2}\\
&< n/4
\end{align*}
with probability $1-e^{-\Omega(n^{\e/6})}$.
So $(1-\e)\log_{1+p}n$ rounds are, with probability $1-e^{-\Omega(n^{\e/6})}$, not enough to complete Phase 2.
\end{proof}

\begin{lemma}\label{RandomPhase3}
Let $\e > 0$. With probability $1- e^{-\Omega(n^\e)}$, we need more than $(1-\e)\frac{1}{p}\ln n$ rounds to complete Phase 3.
\end{lemma}
\begin{proof}
Let $t_2$ denote the number of rounds needed to inform the first $n/4$ nodes. 
This means that we have $n/4 \leq |I_{t_2}| < n/2$.

For the remainder of the proof we will consider a modified model in which \emph{every} node (not only informed ones) randomly chooses a neighbor at the beginning of each round, and if this neighbor was uninformed it will then be considered informed independently with probability~$p \in (0,1]$. A lower bound for the broadcast time of the modified model also is a lower bound for the broadcast time of the original model. 

Enumerate the uninformed nodes at time $t_2$ from 1 to $|V\setminus I_{t_2}|$. 
Define the indicator random variables $X_1,\dots , X_{|V\setminus I_{t_2}|}$ such that for $i \in \{1, \dots, |V\setminus I_{t_2}|\}$, we have

\begin{equation*}
X_i=\begin{cases}
     1 &\mbox{if node }i\mbox{ is uninformed at time }t_2+(1-\e)\frac{1}{p}\ln n,\\
     0 &\mbox{otherwise.}
    \end{cases}
\end{equation*}

They are negatively correlated, as in the modified model the following holds. For any uninformed node $i \in \{1, \dots, |V\setminus I_{t_2}|\}$, the information that other uninformed nodes become informed during rounds $t_2+1$ to $t_2+(1-\e)\frac{1}{p}\ln n$ makes it more likely for $i$ to remain uninformed, and the information that other uninformed nodes remain uninformed makes it more likely for $i$ to become informed.

Let $X:=\sum_{i=1}^{|V\setminus I_{t_2}|} X_i$. 
This is the number of uninformed vertices at time $t_2+(1-\e)\frac{1}{p}\ln n$. 
Since $|V\setminus I_{t_2}| = n - |I_{t_2}| \geq n/2$ and because for any uninformed node $i \in \{1, \dots, |V\setminus I_{t_2}|\}$ we have
\begin{equation*}
\Pr(X_i=1)\geq \left(1-\frac{p}{n-1}	\right)^{(n-1)(1-\e)\frac{1}{p}\ln n},
\end{equation*}
we can bound the expected value as follows.
\begin{align*}
 \E(X) &= \sum_{i=1}^{|V\setminus I_{t_2}|} \Pr(X_i=1)
\geq \frac{n}{2}\left(1-\frac{p}{n-1}	\right)^{(n-1)(1-\e)\frac{1}{p}\ln n}.
\end{align*}
Since for small enough $x > 0$ we have $1 - x \geq e^{-x-x^2}$, we obtain
\begin{align*}
 \E(X) 
&\geq \frac{n}{2}\exp\left(-(1-\e)\ln n-\frac{p}{n-1}(1-\e)\ln n\right)\\
&\geq \tfrac{1}{4}ne^{-(1-\e)\ln n}\\
&= \tfrac{1}{4}n^\e,
\end{align*}
assuming that $n$ is sufficiently large.
By \thmref{chernoffnc} we get
\begin{equation*}
\Pr(X=0)\leq \Pr\left(X\leq (1-\tfrac{1}{2})\E(X)\right)\leq e^{-\E(X)/8} \leq e^{-n^\e/32}.
\end{equation*}

So with probability at least $1 - e^{-n^\e/32}$, we see that $(1-\e)\frac{1}{p}\ln n$ rounds are insufficient to complete Phase 3.
\end{proof}

\begin{proof}[Proof of Theorem 1]
 Follows immediately from Lemmas 2 and 3.
\end{proof}

\section{Upper Bound for Quasirandom Rumor Spreading}\label{sec:upper}

In this section we analyze the quasirandom counterpart of the rumor spreading model described in the previous section.
This model differs from the random model in that each vertex is equipped with a cyclic list of its neighbors and only chooses the first neighbor it attempts to contact at random. After its initial choice, each vertex subsequently attempts to contact the remaining vertices in the order of its list.
Each of these attempts is independently successful with probability $p$. Note that we do not assume that the sender is notified of a transmission failure.

Our goal is to prove that the rumor in the quasirandom model spreads at least as quickly as in the random model. To this aim, we prove the following.

\begin{theorem}\label{thm:upper}
For every $\e>0$ and~$p \in (0,1]$, the number of rounds we need to inform all the nodes of the complete graph on $n$ vertices using the quasirandom rumor spreading model with message success probability $p$ is at most
\begin{equation*}
(1+\e)\left(\log_{1+p}n+\tfrac{1}{p}\ln n\right)
\end{equation*}
with probability at least $1 - n^{-p\e/40}$.
\end{theorem}

Unfortunately, since the rumor spreading process is saturated with many dependencies, determining the runtime for the the quasirandom model is not straightforward. As in~\cite{DFS08}, we try to overcome this difficulty by suitably simplifying the random experiment, in particular, by assuming that certain vertices stop informing (\emph{ignoring}), and that other vertices do not immediately start their own informing process after becoming informed (\emph{delaying}). 
Delaying turns out to be useful as it gives us some influence on when a vertex uses its one random choice. Nodes that have been informed but have not yet begun informing new nodes play an important role in our analysis. We call them \emph{newly informed} vertices.

To obtain bounds that are precise up to the leading constant, however, we have to be careful that our delaying and ignoring techniques do not slow down the rumor spreading process too much. For this reason, we partition the set of rounds that are necessary to inform all the nodes in the graph into two different types of phases. For both types of phases, the set of nodes that are initially active is the set of newly informed nodes.
	
\emph{Lazy phases} were also used in the time analysis of~\cite{DFS08}. Only nodes that are considered active at the beginning of the phase are considered active for the remainder of the phase.
Nodes that are contacted during the phase, although they are still considered to be informed, remain inactive, and are therefore unable to spread the rumor themselves for the continuation of the phase.

Since lazy phases neglect the rumor spreading potential of a significant portion of the nodes, we also need \emph{busy phases}. 
Here, all nodes informed during the busy phase are active for the remainder of the phase. 
In other words, nodes newly informed during the busy phase have the ability to spread the rumor in each subsequent round until the termination of the phase. By choosing the lengths of the busy phases suitably, we balance the difficulties with the inherent dependencies and the losses due to ignoring informed vertices at the end of each phase.
	
As a result of implementing phases in which vertices that can spread the rumor in the original model are now inactive, we are only delaying the point in time at which all the vertices are informed.
Therefore, the upper bound for the quasirandom model with lazy and busy phases holds as an upper bound for the original quasirandom model.

We will split the rumor spreading process into lazy and busy phases in the following way. We start with two lazy phases of $\frac{1}{2}\e\ln n$ rounds each. The main purpose of these two phases, which are easy to analyze, is to inform a set of vertices that is sufficiently large enough to maximize the effectiveness of the subsequent busy phases. We then perform a logarithmic number of busy phases, each composed of a constant number of rounds. This process results in a constant fraction of informed nodes, and we only need two more lazy phases to render the entire network informed.

Let $I_t$ denote the set of vertices that are informed at a given time-step $t$. Similarly, we will denote the set of newly informed vertices at time $t$ by $N_t$.

\subsection{The First Lazy Phase}

The first lazy phase lasts for $\frac{1}{2}\e\ln n$ rounds. 
Our goal is to prove the following.

\begin{lemma}\label{LazyPhase1}
Let $\e>0$. After one lazy phase of length $\frac{1}{2}\e\ln n$, at least $\frac{1}{3}p\e\ln n$ nodes are newly informed with probability at least $1-n^{-p\e/36}$.
\end{lemma}
\begin{proof}
Let $t_1:=\frac{1}{2}\e\ln n$. At time $t=0$ one node, $v_0$, is informed. We perform a lazy phase of length $t_1$. This means that $v_0$ contacts each of the first $t_1$ nodes from its list with probability $p$. Therefore, 
 
\begin{equation*}\E (|N_{t_1}|)=\tfrac{1}{2}p\e\ln n.
\end{equation*}
Using Chernoff bounds we see that

\begin{align*}
 \Pr\left(|N_{t_1}|< \tfrac{1}{3}p\e\ln n	\right)&= \Pr\left(|N_{t_1}|< \left(1-\tfrac{1}{3}\right)\E (|N_{t_1}|)\right)\\
& \leq \exp\left(-\E (|N_{t_1}|)/18	\right)\\
&= \exp\left(-p\e\ln n/36	\right)\\
&= n^{-p\e/36}.
\end{align*}
\end{proof}

\subsection{The Second Lazy Phase}

The second lazy phase begins at time $t_1+1$ and terminates after $\frac{1}{2}\e\ln n$ rounds.
Our goal is to prove the following.

\begin{lemma}\label{LazyPhase2}
Let $\e>0$. If, at some point $t_1$ in our model, we have $\frac{1}{3}p\e\ln n \leq |N_{t_1}| \leq \frac{1}{2}\e\ln n$ and $|I_{t_1}|\leq \frac{1}{2}\e\ln n + 1$, then after one lazy phase of length $\frac{1}{2}\e\ln n$, at least $\left(\frac{1}{3}p\e\ln n\right)^2$ nodes are newly informed with probability at least $1-n^{-\gamma}$ for any $\gamma \in [0, 1)$.
\end{lemma}
\begin{proof}
Let $t_1\in \N$ be such that $\frac{1}{3}p\e\ln n \leq |N_{t_1}| \leq \frac{1}{2}\e\ln n$ and that $|I_{t_1}|\leq \frac{1}{2}\e\ln n + 1$ and let $t_2:=t_1+\frac{1}{2}\e\ln n$. Enumerate the nodes of $N_{t_1}$ from 1 to $|N_{t_1}|$, and impose an artificial ordering on the set so that each node $i$ calls $\frac{1}{2}\e\ln n$ of its neighbors, determined from its cyclic list and its initial random decision, before node $i+1$ attempts any contact.
For each $i\in \{1,\dots,|N_{t_1}|\}$, let $U_i$ denote the set of vertices that $i$ attempts to contact during the next $\frac{1}{2}\e\ln n$ rounds and let $X_i$ be the indicator random variable of the event that $U_i$ is disjoint from $\left(\bigcup_{j=1}^{i-1}U_j\right)\cup I_{t_1}$.
If $X_i= 1$ for all $i\in \{1,\dots,|N_{t_1}|\}$, then $|N_{t_2}|$ is equal to the number of contacts made during this phase. 

When vertex $i$ first attempts contact, at most $\left|I_{t_1} \setminus \{i\}\right| + \frac{1}{2}(i - 1)\e\ln n \leq \frac{1}{2}i\e\ln n$ other vertices are already informed. The probability for $i$ to attempt to contact one of these vertices is largest when they are at distance at least $\frac{1}{2}\e\ln n$ from each other in the list of $i$. Therefore, 
\begin{equation}
\Pr(X_i=0)\leq \frac{\left(\tfrac{1}{2}\e\ln n\right)\left(\tfrac{1}{2}i\e\ln n\right)}{n-1}\leq \frac{\left(\tfrac{1}{2}\e\ln n\right)^3}{n-1}.
\end{equation}
Using a simple union bound, we conclude that
\begin{align*}
\Pr(\forall i\in\{1,\dots,|N_{t_1}|\}:X_i=1)& = 1-\Pr(\exists i\in\{1,\dots,|N_{t_1}|\}:X_i=0)\\
&\geq 1-\sum_{i=1}^{|N_{t_1}|}\Pr(X_i=0)\\
&\geq 1-\frac{\left(\tfrac{1}{2}\e\ln n\right)^4}{n-1}.
\end{align*}

Now that we have shown that the chances of contacting an already informed vertex in this phase are sufficiently small, all that is left to do is to determine how many contacts are made during the phase.

Every node in $N_{t_1}$ attempts to contact $\frac{1}{2}\e\ln n$ nodes, so there are $\frac{1}{2}\e\ln n|N_{t_1}|$ possible contacts made during the phase. Each of these is independently successful with probability $p$. Let $Y$ be the random variable denoting the number of contacts that are actually made during the phase. Then we have
\begin{equation}
\E(Y)
=	\tfrac{1}{2}p\e\ln n|N_{t_1}|
\geq	\left(\tfrac{1}{2}p\e\ln n\right) \left(\tfrac{1}{3}p\e\ln n\right) .
\end{equation}
Using Chernoff bounds, we see that

\begin{align*}
\Pr\left(Y<\left(\tfrac{1}{3}p\e\ln n	\right)^2	\right)&\leq\Pr\left(Y<\left(1-\tfrac{1}{3}	\right)\E(Y)\right)\\
&\leq e^{-\E(Y)/18}\\
&\leq e^{-(p\e\ln n)^2/108}\\
&=n^{-p^2\e^2(\ln n)/108}.
\end{align*}

Therefore, at least $\left(\frac{1}{3}p\e\ln n	\right)^2$ vertices are informed during this phase with probability at least $1-\frac{\left(\frac{1}{2}\e\ln n\right)^4}{n-1}-n^{-p^2\e^2(\ln n)/108}
\geq	1-n^{-\gamma}$ for any fixed $\gamma \in [0, 1)$.
\end{proof}


\subsection{The Busy Phases}
A sufficient number of nodes are informed of the rumor in the first lazy phase, and so we are ready to commence the set of busy phases. As we have mentioned earlier, the idea of these phases is that
nodes informed during each busy phase are able to spread the rumor during subsequent rounds of this phase. Because of the dependencies, these phases require a more refined analysis. Our goal is to inform a constant fraction of the nodes in the network by the time we complete this sequence of phases.

\subsubsection{The Analysis of a Single Busy Phase.}
In order to determine the cumulative effect of the busy phases, we must first analyze the impact of a single busy phase composed of $k$ rounds starting after time-step $t$. The theorem we present below is the heart of the precise analysis of the quasirandom model. The idea of the proof is to investigate the part of the process originating from each single node in $N_{t}$. A single such process can be analyzed with moderate difficulty. Unfortunately, there may be  ``conflicts'' among these partial processes, that is, several of these partial processes may inform the same node, possibly at different times. However, we show that only few of these conflicts occur. By completely ignoring all parts that are contained in a conflict, we manage to analyze the busy phase.

Let $\e' > 0$, $k \in \N$, $p \in (0,1]$ and~$\zeta' \leq \frac{2^{-k}}{k}(2e)^{-\frac{2^{k-1}}{p^3(1+p)^{k-3}} - k - 1}.$ We will prove the following statement.

\begin{theorem}\label{1busyphase}
Let $t \in \N$ such that in our model at point $t$ we have $|N_{t}|\geq (p\e'\ln n)^2$ and $|I_{t}| \leq \zeta' n$.
For any $c >0$,
if we perform a busy phase of length $k$, then at the conclusion of this busy phase, the number of newly informed vertices satisfies with probability at least $1 - n^{-c}$ the inequality
\begin{equation*}|N_{t+k}|\geq p(1+p)^{k-2}|N_t|.\end{equation*}
\end{theorem}

\begin{proof}
Let $\zeta  := 2^{k}\zeta'.$
Let $t$ be such that we have $|N_{t}|\geq (p\e'\ln n)^2$ and $|I_{t}| \leq \zeta' n$. Note that the number of informed nodes at time $t+k$ can not exceed $\zeta n$, and this even if we consider a failure-free process between time $t$ and $t+k$ which we will do in the following.

Enumerate the nodes of $N_t$ from 1 to $|N_t|$. 
For each $i \in \{1, \dots, |N_t|\}$, we define the set of \emph{potential descendants of} $i$, denoted $P^{(k)}_i$, as the set of nodes which would be directly or indirectly informed within rounds $t+1,\dots,t+k$ by $i$ if no failures occured and only uninformed nodes were targeted. More general, for any node $v\in V$ and any $j \in [n-1]$ the set $P^{(j)}_v$ is recursively defined as follows.
Let $v_1$ be the randomly chosen neighbor of $v$ which $v$ attempts to contact in the round after being informed, and let $v_1,\dots,v_{n-1}$ be the list of $v$. Then
\begin{align*}
P^{(1)}_v &:= \{v_1\} \mbox{, and, for}\ \ j \in [n-2],\\
P^{(j+1)}_v &:= \{v_1,\dots,v_{j+1}\} \cup P^{(j)}_{v_1} \cup \dots \cup P^{(1)}_{v_j}.
\end{align*}
We say that $i$ is \emph{conflict-free} if the following two conditions hold.
	\begin{enumerate}
	\item $|P^{(k)}_i| = 2^k-1$ and
	\item $P^{(k)}_i\cap \left(P^{(k)}_1 \cup \dots \cup P^{(k)}_{i-1} \cup I_{t}\right) = \emptyset$. 
	\end{enumerate}	
Otherwise, we call $i$ \emph{conflicting}.

\begin{claim}\label{cl:1}
For any $i \in \{1, \dots , |N_t|\}$ and any set $M \subseteq [i-1]$
$$\Pr\left( i \mbox { is conflicting}\ |\ \ \forall m \in M: \ \ m \mbox { is conflicting}\right) \leq 2^{k+1} \zeta k .$$
\end{claim}

\begin{proof}
To bound the probability that the first condition of conflict-freeness fails, we impose an ordering on the random decisions of the vertices in $P^{(k)}_i\cup\{i\}$.
For every such decision $d$, the probability that $d$ creates a conflict with any previous decision, i.e., that the node in question attempts to contact a node that is targeted by a different node in $P^{(k)}_i\cup\{i\}$, is bounded from above by $\frac{k}{n-1}2^k$. 
So the probability that among all decisions a conflict is created is bounded from above by $\frac{k}{n-1}(2^k)^2$.

The probability that the second condition of conflict-freeness fails is
\begin{align*}
\Pr&\left(P^{(k)}_i\cap \left(P^{(k)}_1 \cup \dots \cup P^{(k)}_{i-1} \cup I_{t}\right) \neq \emptyset \right).
\end{align*}
Let $d$ be a random decision of a vertex in $P^{(k)}_i\cup\{i\}$.

 For all outcomes of random decisions up to round $t+k$ other than those of vertices in $P^{(k)}_i\cup\{i\}$, the probability that $d$ creates a conflict with any such decision can be \emph{uniformly} bounded from above by $\left(\zeta n - 2^k\right)\frac{k}{n-1}$. So 
\begin{align*}
\Pr&\left(P^{(k)}_i\cap \left(P^{(k)}_1 \cup \dots \cup P^{(k)}_{i-1} \cup I_{t}\right) \neq \emptyset \right)\\
&\leq \left|P^{(k)}_i\cup\{i\}\right|\cdot\left(\zeta n - 2^k\right)\tfrac{k}{n-1}\\
&\leq 2^k \left(\zeta n - 2^k\right)\tfrac{k}{n-1}.
\end{align*}

Using a union bound, the probability that $i$ is conflicting is bounded from above by
$
\frac{k}{n-1}(2^k)^2 + 2^k \left(\zeta n - 2^k\right)\frac{k}{n-1} 
= 2^k \zeta n \frac{k}{n-1}
\leq 2^{k+1} \zeta k.
$ 
\end{proof}

\begin{claim}\label{cl:2}
For any set $M \subseteq \{1, \dots , |N_t|\}$
$$ \Pr\left(\forall i \in M: \ \ i\mbox{ is conflicting}\right)
\leq \left(2^{k+1}k\zeta\right)^{|M|}.$$
\end{claim}

\begin{proof}
Let $M = \{m_1, \dots, m_{|M|}\}$ where $m_1 < \dots < m_{|M|}$. Then by \clref{1} one has
\begin{align*}
 &\Pr\left(m_1, \dots, m_{|M|}\mbox{ are conflicting}\right)\\
= &\prod_{i =1}^{|M|}\Pr\left(m_i\mbox{ is conflicting}\ |\ \  m_1, \dots, m_{i-1}\mbox{ are conflicting}\right)\\
\leq &\left(2^{k+1}k\zeta\right)^{|M|}.
\end{align*}

\end{proof}

Let $\oN_t :=\menge{i \in \{1, \dots , |N_t|\}}{i \mbox { is conflict-free}}.$
We bound the number of conflicting vertices from above using \clref{2}.
Let $q=\frac{p^3(1+p)^{k-3}}{2^{k-1}}$. Then the probability that there are at least $q|N_t|$ conflicting vertices is

\begin{align*}
&\sum_{\begin{subarray}
            1M\subseteq N_t \\ |M|\geq q|N_t|
           \end{subarray}} \Pr\left(\forall i \in M: \ \ i\mbox{ is conflicting}\right)
\ \leq \sum_{\begin{subarray}
            1M\subseteq N_t \\ |M|\geq q|N_t|
           \end{subarray}} \left(2^{k+1}k\zeta\right)^{|M|}\\ 
& \leq\ 2^{|N_t|}\left(2^{k+1}k\zeta	\right)^{q|N_t|}\
= \left(2\left(2^{k+1}k\zeta	\right)^{\frac{p^3(1+p)^{k-3}}{2^{k-1}}}	\right)^{|N_t|}\
\leq\ e^{-|N_t|}\
\leq\ n^{-c}
\end{align*}
for any $c>0$.

This means that we have
\begin{equation}\label{conflictfreeness}
|\oN_t| \geq (1-q)|N_t|
\end{equation}
with probability at least $1 - n^{-c}$.

We will now reconsider the actual, defective process. We will condition on \eqref{conflictfreeness} for the remainder of the proof. This means that we consider the random process split into two independent parts: First we run the failure-free process up to round $t+k$, then for every transmission in the failure-free process the biased coin is flipped to decide if the transmission takes place in the defective process.

Let the set of 
\emph{actual descendants of} $i$, denoted $D^{(k)}_i$, be the set of nodes which are directly or indirectly contacted within rounds $t+1,\dots,t+k$ by $i$ in the defective process. More precisely, any node $v \in P^{(k)}_i$ is in $D^{(k)}_i$
if all the contacting attempts on the path from $i$ to $v$ which the information
takes in the failure-free process are actually successful.

Let $X_i \subseteq D^{(k)}_i$ denote the set of vertices that are actual descendants of $i$ and are contacted in round $t+k$. 

We can say the following about the expectation of $X_i$ if $i$ is conflict-free.
\begin{claim}
For $i \in \oN_t$ we have
\begin{equation*}
\E\left(\left|X_i\right|\right)=p(1+p)^{k-1}.
\end{equation*}
\end{claim}

\begin{proof}
We prove this claim by induction on the number of rounds that have occurred. Assume that $i$ is conflict-free.
At time $t+1$, the probability that $i$ has successfully contacted a new node is $p$, and therefore the expected number of nodes informed by $i$ is $p$.
Let $j \in \{1, \dots ,k-1\}$ be such that for all $r \in \{1, \dots ,j\}$ the expected number of actual descendants of $i$ informed in round $t+r$ is $p(1+p)^{r-1}$.

Then the expected number of actual descendants of $i$ informed in round $t+j+1$ is 
	\begin{align*}
	p\left(1+ \sum_{r=1}^jp(1+p)^{r-1}\right)=p(1+p)^j.
	\end{align*}
\end{proof}
Let $X:=\sum_{i \in \oN_t}\left|X_i\right|$. 
Note that the random variables $|X_i|,  i \in \oN_t$ are independent and bounded by $2^{k-1}$. So we can use Hoeffding bounds (\thmref{Hoeffding}). As $q \leq \frac{p}{1+p}$ and $\E(X) = p(1+p)^{k-1}|\oN_t| \geq p(1+p)^{k-1}(1-q)|N_t|$, we have
\begin{align*}
\Pr \left(X < p(1+p)^{k-2}|N_t|\right)
&\leq	\Pr \left(X < \frac{\E(X)}{(1+p)(1-q)}\right)\\
&\leq 	\exp\left(-\frac{2(p-q-pq)^2\E(X)^2}{2^{2k-2}(1+p)^2(1-q)^2|\oN_t|}\right)\\
&\leq	\exp\left(-\frac{(p-q-pq)^2p^2(1+p)^{2k-4}|N_t|}{2^{2k-3}}\right)\\
&\leq n^{-c}
\end{align*}
for any $c>0$.

So with probability at least $1 - n^{-c}$,
\begin{align*}
|N_{t+k}| &\geq X \geq p(1+p)^{k-2}|N_t|.
\end{align*}
This proves Theorem \ref{1busyphase}.
\end{proof}

We also have to ensure that after the performance of one busy phase, a big enough fraction of the informed vertices is newly informed, since only the newly informed vertices are active in the next phase. We show that this holds in the proof of the following corollary.

\begin{cor}\label{OneBusyPhase2}
Define $t$ such that at point $t$ we have $|N_{t}|\geq (p\e'\ln n)^2$ and $|I_{t}|\leq \min\left\{\zeta' n , \frac{2^k-1}{p(1+p)^{k-2}-1}|N_{t}|\right\}$.
If we perform a busy phase of length $k$, then for any $c >0$ at the conclusion of this busy phase we have
\begin{equation*}
 |I_{t+k}|\leq \frac{2^k-1}{p(1+p)^{k-2}-1}|N_{t+k}|
\end{equation*}
with probability at least $1 - n^{-c}$.
\end{cor}
\begin{proof}
If we perform a busy phase of length $k$, by Theorem \ref{1busyphase} we get
\begin{align*}
 |I_{t+k}| &= |I_t|+\sum_{i=1}^k|N_{t+i}|\\
&\leq \frac{2^k-1}{p(1+p)^{k-2}-1}|N_{t}|+\sum_{i=1}^k2^{i-1}|N_t| \\
& = \frac{(2^k-1)p(1+p)^{k-2}}{p(1+p)^{k-2}-1}|N_t|\\
&\leq \frac{2^k-1}{p(1+p)^{k-2}-1}|N_{t+k}|
\end{align*}
with probability at least $1 - n^{-c}$ for any $c >0$.
\end{proof}

\subsubsection{Assembling of the Busy Phases.}

Now that we have analyzed a single busy phase, we can put these phases together to obtain a constant fraction of informed nodes.
Let $\e > 0$,~$p \in (0,1]$ and $$k := \frac{1+\e}{\e}\left(\log_{1+p}\frac{1}{p} + 2\right).$$
As in the previous section, let
$\zeta \leq \frac{1}{k}(2e)^{-\frac{2^{k-1}}{p^3(1+p)^{k-3}} - k - 1},$ 
and $\zeta' := 2^{-k}\zeta$. 
We show the following.

\begin{theorem}\label{allbusyphases}
Let $\e' > 0$. Let $t_2$ be such that in our model at point $t_2$ we have $|N_{t_2}|\geq (p\e'\ln n)^2$ and $|I_{t_2}|\leq \min\left\{\zeta' n , \frac{2^k-1}{p(1+p)^{k-2}-1}|N_{t_2}|\right\}$. 

Let $\ell$ denote the smallest integer such that if we perform $\ell$ busy phases with $k$ rounds we have $|I_{t_2+\ell k}| \geq \zeta' n.$

Then $$\ell \leq \frac{(1+\e)\log_{1+p}n}{k}$$
and
\begin{equation*}
|I_{t_2 + \ell k}| \leq \frac{2^k-1}{p(1+p)^{k-2}-1}|N_{t_2 + \ell k}|
\end{equation*}
hold with probability at least $1 - n^{-c}$ for any $c >0$.
\end{theorem}

\begin{proof} 
Since $p(1+p)^{k-2} \geq 1$ and $|N_{t_2}| \geq (p\e'\ln n)^2$, we can use Theorem \ref{1busyphase} inductively and obtain that for all $ s\in\{1,\dots,\ell\}$ we have

\begin{equation*}
|N_{t_2+sk}|\geq p(1+p)^{k-2}|N_{t_2+(s-1)k}|\geq\dots\geq  \left(p(1+p)^{k-2}\right)^s	|N_{t_2}|
\end{equation*}
with probability at least $1 - sn^{-c}$ for any $c >0$, and therefore

\begin{align*}
|I_{t_2+\ell k}|&\geq \left(p(1+p)^{k-2}\right)^{\ell},
\end{align*}
with probability at least $1 - \ell n^{-c}$ for any $c >0$.
Since $ n \geq \zeta n \geq |I_{t_2 + \ell k}|$, we have
\begin{align*}
\log_{1+p}n
&\geq 	\log_{1+p}\zeta n\\
&> \log_{1+p}\left(p(1+p)^{k-2}\right)^{\ell}\\
&=\ell(k-2+\log_{1+p}p),
\end{align*}
which implies that
$$
{\ell}\leq \frac{\log_{1+p}n}{k-2+\log_{1+p}p} = \frac{(1+\e)\log_{1+p}n}{k}$$
with probability at least $1 - n^{-c}$ for any $c >0$.

Since
$
|I_{t_2}| \leq \frac{2^k-1}{p(1+p)^{k-2}-1}|N_{t_2}|,
$
by an inductive application of Corollary \ref{OneBusyPhase2}, we get that 
\begin{equation*}
|I_{t_2 + \ell k}| \leq \frac{2^k-1}{p(1+p)^{k-2}-1}|N_{t_2 + \ell k}|
\end{equation*}
holds with probability at least $1 - n^{-c}$ for any $c >0$.
\end{proof}

\subsection{Second To Last Phase}

Now that we have a small constant fraction of newly informed nodes, a lazy phase of a constant number of rounds suffices to yield a large fraction of newly informed nodes.

\begin{lemma}\label{LazyPhase3}
Let $\e \in (0,1)$ and $k := \frac{1+\e}{\e}\left(\log_{1+p}\frac{1}{p} + 2\right)$.
Let $t_3$ be such that in our model at round $t_3$ we have $
|I_{t_3}| \leq \frac{2^k-1}{p(1+p)^{k-2}-1}|N_{t_3}|,$
and that there exist $\zeta, \zeta' \in (0,1)$ such that
 $\zeta'n \leq |I_{t_3}| \leq \zeta n$ holds.
Let $S:=\frac{2^k\ln(1/\zeta)}{p\zeta'}$.

 After one lazy phase of $S$ rounds starting at time $t_3$, at least $\left(1-3\zeta\right)n$ nodes will be newly informed with probability $1-e^{-\Omega(n)}$.
\end{lemma}
\begin{proof}
We perform one lazy phase of $S$ rounds starting at time $t_3$. Let $v_0\in V\setminus I_{t_3}$. Then

\begin{align*}
\Pr(\mbox{no } &v\in N_{t_3}\mbox{ contacts }v_0\mbox{ in this phase})\ =\ \left(1-\frac{p}{n-1}\right)^{S|N_{t_3}|}\\
&\leq\ \exp\left(-\frac{pS|N_{t_3}|}{n-1}\right)\ 
\leq\ \exp\left(-\frac{Sp(p(1+p)^{k-2}-1)|I_{t_3}|}{2^k(n-1)}\right)\\
&\leq\ \exp\left(-\frac{Sp|I_{t_3}|}{2^k(n-1)}\right)\ 
\leq\ \exp\left(-\frac{Sp\zeta'}{2^k}\right)\
=\ \zeta.
\end{align*}
We now calculate the expected number of newly informed nodes after $S$ rounds. With $t_4:=t_3+S$ we have

\begin{align*}
\E\left(|N_{t_4}|\right) &= |V\setminus I_{t_3}|\cdot \Pr(v_0\in V\setminus I_{t_3}\mbox{ is informed in this phase})\\
&\geq (n-|I_{t_3}|)\left(1-\zeta\right)\\
&\geq \left(n-\zeta n\right)\left(1-\zeta\right)\\
&\geq (1-2\zeta)n.
\end{align*}

We will now use the Azuma-Inequality (see Lemma
\ref{thm:azuma}). Number the nodes of $N_{t_3}$ from 1 to $|N_{t_3}|$. 
Then for all $i\in\{1,\dots,|N_{t_3}|\}$, define the random variable $X_i$ as the set of vertices that $i$ contacts in the $S$ lazy rounds. 
Now we can define the function $f$ such that

\begin{equation*}
f(X_1,\dots,X_{|N_{t_3}|}) := \left|\bigcup_{i=1}^{|N_{t_3}|}X_i\setminus I_{t_3}	\right|=|N_{t_4}|.
\end{equation*}

By this definition, we see that

\begin{equation*}
f(x_1,\dots,x_i,\dots,x_{|N_{t_3}|})-f(x_1,\dots,x_i',\dots,x_{|N_{t_3}|})\leq S.
\end{equation*}

Therefore, we can calculate the probability that we inform less than $(1-3\zeta)n$ vertices in this phase.

\begin{align*}
\Pr&\left(|N_{t_4}|<\left(1-3\zeta\right)n\right)\ =\ \Pr\left(|N_{t_4}|<(1-2\zeta)n-\zeta n\right)\\
&\leq \Pr\left(\left|N_{t_4}-\E(|N_{t_4}|)\right|\geq	\zeta n\right)
\ \leq\ 2\exp\left(-\frac{2\zeta^2n^2}{\sum_{i=1}^{|N_{t_3}|}S^2}\right)\\
&\leq 2\exp\left(-\frac{2\zeta^2n^2}{\zeta n S^2}\right)
\ =\ e^{-\Omega(n)}.
\end{align*}
\end{proof}


\subsection{The Final Phase}

The last phase of the protocol is again a lazy phase. We now use the large fraction of newly informed nodes from the previous phase to inform the few remaining nodes.
\begin{lemma}\label{LazyPhase4}
Let $\e \in (0,1)$ and $\eta \leq \frac{\e}{4}$.
Let $t_4$ be such that in our model at round $t_4$ we have $
|N_{t_4}| \geq \left(1-\eta\right)n.$
 After one lazy phase of $\frac{(3+\e)}{3p}\ln n$ rounds starting at time $t_4$, all the nodes will be informed with probability $1-O(n^{-\e(1-\e)/12})$.
\end{lemma}
\begin{proof}
We will perform one lazy phase of $\frac{(3+\e)}{3p}\ln n$ rounds starting at time $t_4$.
Let $v_0\in V\setminus I_{t_4}$. Then
\begin{align*}
\Pr(\mbox{no } &v\in N_{t_4}\mbox{ contacts }v_0\mbox{ in this phase})\ 
=\ \left(1-\frac{p} {n-1} \right)^{(3+\e)|N_{t_4}|\ln n/3p}\\
&\leq\ \exp\left(\frac{-(3+\e)\ln n|N_{t_4}|}{3(n-1)}\right)
\ \leq\ \exp\left(\frac{-(3+\e)(1-\eta)n\ln n}{3(n-1)}\right)\\
&\leq\ \exp\left(-\frac{(3+\e)(1-\eta)\ln n}{3}\right)\ 
\leq\ n^{-\left(1+\e(1-\e)/12\right)}.
\end{align*}
So the probability that all the nodes become informed is
\begin{align*}
\Pr(\forall &v\in V\setminus I_{t_4}:\ v \mbox{ becomes informed})\\&=1-\Pr(\exists v\in V\setminus I_{t_4}:\ v\mbox{ does not get informed})\\
&\geq 1-\sum_{v\in V\setminus I_{t_4}}\Pr(v\mbox{ does not get informed})\\
&\geq 1-\eta n n^{-\left(1+\e(1-\e)/12\right)}\
=\ 1-O(n^{-\e(1-\e)/12}).
\end{align*} 
\end{proof}

\subsection{Proof of Theorem \ref{thm:upper}}

Let $\e \in (0,1)$,~$p \in (0,1]$ and $k := \frac{1+\e}{\e}\left(\log_{1+p}\frac{1}{p} + 2\right)$.
Furthermore, let 
$\zeta := \min\left\{ \frac{1}{k}(2e)^{-\frac{2^{k-1}}{p^3(1+p)^{k-3}} - k - 1}, \frac{\e}{12}\right\}$ and 
$\zeta' := 2^{-k}\zeta.$

We start a delayed quasirandom rumor spreading protocol with message success probability $p$ and with one initially informed vertex. We first perform one lazy phase of length $t_1:= \frac{1}{2}\e\ln n$. By Lemma \ref{LazyPhase1} this yields that
$|N_{t_1}| \geq \frac{1}{3}p\e\ln n$ holds with probability at least $1-n^{-p\e/36}$. 
Of course, after one lazy phase of length $t_1$ we have with probability one $|N_{t_1}| \leq t_1$ and $|I_{t_1}| \leq t_1 + 1$. So we can apply Lemma \ref{LazyPhase2} and get $|N_{t_2}| \geq \left(\frac{1}{3}p\e\ln n\right)^2$, this phase succeeds with probability at least $1-n^{-\gamma}$ for any $\gamma \in [0, 1)$. 
Furthermore we have with probability one
$|I_{t_2}| \leq \left(\frac{1}{2}\e\ln n\right)^2 + \frac{1}{2}\e\ln n + 1 \leq \zeta' n$
as well as $|I_{t_2}| = |I_{t_1}| + |N_{t_2}| \leq \frac{2^k-1}{p(1+p)^{k-2}-1}|N_{t_2}|$ for any sufficiently large $n$. 
So we can apply Theorem \ref{allbusyphases} with $\e' := \frac{\e}{3}$. This gives us an $\ell \leq \frac{(1+\e)\log_{1+p}n}{k}$ such that if we set $t_3:= t_2 + \ell k$, then for any $c >0$ we have with probability at least $1-n^{-c}$ 
$$\zeta' n \leq |I_{t_3}| \leq \zeta n \ \ \ \ \mbox{   and    }\ \ \ \ |I_{t_3}| \leq \frac{2^k-1}{p(1+p)^{k-2}-1}|N_{t_3}|.$$
So with probability at least $1-n^{-c}$ the preconditions of Lemma \ref{LazyPhase3} are fulfilled. Therefore, if we set $S:=\frac{2^k\ln(1/\zeta)}{p\zeta'}$ and $t_4 := t_3 + S$, we get
$|N_{t_4}| \geq \left(1-3\zeta\right)n$ with probability at least $1-n^{-c}$.
We can consequently apply Lemma \ref{LazyPhase4} with $\eta := 3\zeta$. We conclude that after $\frac{(3+\e)}{3p}\ln n$ more rounds all the nodes will be informed with probability  $1-O(n^{-\e(1-\e)/12})$.

Overall, we perform at most
\begin{equation*}
\tfrac{1}{2}\e\ln n+(1+\e)\log_{1+p} n+ S+\tfrac{3+\e}{3p}\ln n\leq (1+\e)\left(\tfrac{1}{p}\ln n+\log_{1+p} n\right)
\end{equation*}
rounds in our delayed quasirandom rumor spreading protocol with message success probability $p$.

The overall failure probability is at most
$$n^{-p\e/36} + n^{-\gamma} + n^{-c} + e^{-\Omega(n)} + O(n^{-\e(1-\e)/12}) \leq n^{-p\e/40}.$$

\subsection{Upper Bound for the Fully Random Protocol}

The proof of the upper bound provided in Theorem~\ref{thm:upper} can easily be modified to yield the corresponding bound for the classical randomized rumor spreading protocol.

\begin{theorem}\label{thm:upperrandom}
For every $\e>0$ and~$p \in (0,1]$, the number of rounds we need to inform all the nodes of the complete graph on $n$ vertices using the fully random rumor spreading model with transmission failure rate $1-p$ is
\begin{equation*}
 (1+\e)\left(\log_{1+p}n+\tfrac{1}{p}\ln n\right)
\end{equation*} 
with probability at least $1-n^{-p\e/40}$.
\end{theorem}

We omit a formal proof, since all the necessary arguments are used in the proof of Theorem~\ref{thm:upper} and can be reapplied to fit our needs. Much of our reasoning can be simplified significantly since we do not have to cope with the dependencies present in the quasirandom model. 

The only aspect in which the quasirandom model is superior to the random model is the first phase. Here we have the advantage that one vertex does not inform a neighbor more than once in time shorter than its degree. However, for the fully random model the chance that a vertex informs one or more neighbors multiple times within $t$ rounds is only $\Theta(t^2/n)$. Hence the benefits of the quasirandom model are minimal.

\section{Conclusion}

In this paper, we present the first precise results pertaining to the robustness of randomized rumor spreading and its quasirandom variant up to constant factors. We showed that if the network topology is a complete graph on $n$ vertices and each transmission only reaches its destination with probability $p$, then after $(1+\e)\left(\frac{1}{p}\ln n+\log_{1+p} n\right)$ rounds, the quasirandom protocol will have informed all nodes in the graph with probability at least $1-n^{-p\e/40}$. For $p = 1$, this result coincides with the known sharp bound of $(1 + o(1)) (\log_2 n + \ln n)$. We also showed that the robustness of the quasirandom model is at least as good as that of the classical protocol. 

This work prompts the question: do we observe similar robustness behavior for other network topologies? We believe that for most natural network topologies, the quasirandom variant is as robust as the fully random one. Unfortunately, no other precise analysis up to leading constants of the run-time of non-trivial graph classes has been published, not even for the fully random model. Hence, to fully understand robustness, it is necessary to first conceive of a precise analysis of the classical model, e.g., on hypercubes or on random graphs.
}

\bibliographystyle{alpha}
\bibliography{lit}
\end{document}